\documentclass{article}
\usepackage[utf8]{inputenc}
\usepackage{authblk}
\usepackage{amsmath,amsthm,amssymb,color,latexsym}
\usepackage{amssymb}
\usepackage{enumitem}  
\newtheorem{theorem}{Theorem}

\newtheorem{lemma}{Lemma}
\newtheorem{proposition}{Proposition}
\newtheorem{rem}{Remark}
\usepackage{graphicx}
\usepackage{appendix}

\title{On the Joint Evolution Problem\\ for a Scalar Field and its Singularity}
\author[1]{Aditya Agashe\thanks{aditya\_agashe@brown.edu}}
\author[2]{Ethan Lee\thanks{esl75@scarletmail.rutgers.edu}}
\author[2]{Shadi Tahvildar-Zadeh\thanks{shadit@math.rutgers.edu}}
\affil[1]{Brown University}
\affil[2]{Rutgers University}

\date{January 19, 2023}

\begin{document}

\maketitle

\begin{abstract}
In the classical electrodynamics of point charges in vacuum, the electromagnetic field, and therefore the Lorentz force, is ill-defined at the locations of the charges. Kiessling resolved this problem by using the momentum balance between the field and the particles, extracting an equation for the force that is well-defined where the charges are located, so long as the field momentum density is locally integrable in a neighborhood of the charges. 

In this paper, we examine the effects of such a force by analyzing a simplified model in one space dimension. We study the joint evolution of a massless scalar field together with its singularity, which we identify with the trajectory of a particle. The static solution arises in the presence of no incoming radiation, in which case the particle remains at rest forever. We will prove the stability of the static solution for particles with positive bare mass by showing that a pulse of incoming radiation that is compactly supported away from the point charge will result in the particle eventually coming back to rest. We will also prove the  nonlinear instability of the static solution for particles with negative bare mass by showing that an incoming radiation with arbitrarily small amplitude will cause the particle to reach the speed of light in finite time. We conclude by discussing modifications to this simple model that could make it more realistic.
\end{abstract}

\section{Introduction and Main Result}

Classical electromagnetism has a fundamental problem: For a charged point-particle in an electromagnetic field that is at least partly sourced by that particle, the field is not defined at the location of the particle. Because the Lorentz force that the field exerts on the particle depends on the values of the field at the particle's location, the force is also undefined where it's needed, i.e. for the particle's equations of motion to make sense. This is the infamous {\em radiation-reaction problem}. This problem has been the subject of intense study by some of the world's most renowned physicists and mathematicians, including Poincar\'e \cite{Poi1906} and Dirac \cite{Dir38}, for more than a century. An excellent account of this endeavor can be found in \cite{SpohnBook}, where an entire chapter is devoted to recounting its history\footnote{It is outside the scope of this article to mention all the various directions taken by researchers to resolve this issue.  Interested readers are referred to \cite{SpohnBook} and its copious bibliography.}. 

An important breakthrough came in 2019 when, following up on some ideas of Poincar\'e \cite{Poi1906}, Kiessling \cite{kiessling,kiessling-err} showed that if one postulated local conservation laws to hold for the total (field + particle) energy density-momentum density-stress tensor, i.e. (employing the Einstein summation convention, where repeated indices are summed over their range)
\begin{equation}\label{eq:Tconserv}
    \partial^\mu T^{\mbox{\tiny total}}_{\mu\nu} = 0,\quad \mbox{where}\quad T_{\mu\nu}^{\mbox{\tiny total}} = T_{\mu\nu}^{\mbox{\tiny field}} + T_{\mu\nu}^{\mbox{\tiny particle}},
\end{equation}
(once these expressions are properly defined) then the force can be determined using the law of momentum balance, provided the field momentum density is locally integrable in a neighborhood of the charge. This integrability assumption rules out the classical electromagnetic vacuum law $E=D, B=H$ postulated by Maxwell, but admits others, such as the Bopp-Land\'e-Thomas-Podolsky (BLTP for short) vacuum law \cite{Bop40,Bop43,Lan41,LT41,Pod42}.

It is of interest to study the effect of the Kiessling force on the motion of an electromagnetic point-charge. In three space dimensions using the standard electromagnetic vacuum laws, this is not possible. One can use other vacuum laws in three space dimensions, such as BLTP, that do meet Kiessling's criterion, and for which one can prove local well-posedness of the joint field and particle dynamics \cite{kiessling,KieTah23} as well as global existence for the solution to the scattering problem of a single particle by a smooth potential \cite{hoang}. However, the expression for the force in the BLTP case is quite complicated, which makes it hard to figure out what the particle trajectories actually look like.  On the other hand, by a simple scaling analysis, it is easy to see that Kiessling's criterion may be satisfied for Maxwell's vacuum law, so long as one works in {\em one} space dimension.    Since there is however no viable electromagnetism in one space dimension, we instead turn to the simpler model of a {\em scalar} charge.  Such a model has been proposed before by many authors, see e.g. \cite{EKR09} and references therein. To simplify matters even further, we will be focusing on the case of a single particle perturbed by scalar radiation. Following Weyl's ideas on singularity theories of matter \cite{Weyl}, we will take the evolving singularity of the scalar field to represent the path of the particle in space-time. This is thus a joint evolution problem for a scalar field $u(t,s)$, and the trajectory of its singularity, $s = q(t)$. The governing equations are as follows (see \cite{EKR09}, Eqs. 7--9): The field satisfies

\begin{equation}\label{eq:wave}
    \begin{cases} 
      \partial_{t}^2 u - \partial_{s}^2 u = a\delta(s - q(t)) \\
      u(0, s) = -\frac{a}{2}\lvert s\rvert + V_0(s)\\
      \partial_{t}u(0, s) = V_1(s) \\
    \end{cases}
\end{equation}
($\delta$ is the Dirac delta-function) while the equations of motion for the particle are
\begin{equation}\label{eq:eom}
   \begin{cases} 
      \dot{q} = \frac{p}{m\sqrt{1 + \frac{p^2}{m^2}}} \\
      \dot{p} = f(t,q,\dot{q}) \\
      q(0) = 0 \\
      \dot{q}(0) = 0
   \end{cases}
\end{equation}
(The speed of light has been set equal to one.)

Here, (\ref{eq:wave}) is the Cauchy problem for a massless wave equation sourced by the particle. We have added $V_{0}(s)$ and $V_{1}(s)$ to the initial data to represent smooth incoming radiation that is compactly supported away from the point charge. Thus $V_{0}, V_{1} \in C^{\infty}_{c}(\mathbb{R} \setminus \{0\})$. Real constants $a$ and $m$ represent the charge and the (bare) rest mass of the particle. Equations (\ref{eq:eom}) are simply Newton's equations of motion with Einstein's special-relativistic relation between momentum and velocity instead of Newton's. $f$ is the force exerted by the field on the particle, and its precise expression needs to be determined using another principle.  Here, following Kiessling, we will use {\em momentum conservation} to determine $f$.
\begin{rem}
We note that the above system of equations is not fully Lorentz-covariant: the right-hand side of the wave equation in \eqref{eq:wave} does {\em not} transform correctly under a Lorentz transformation.  This is a defect of the model (which was also pointed out in \cite{EKR09}).  This defect can be corrected, but the resulting system becomes harder to analyze.  Some of the results in this paper have also been obtained for the fully relativistic version, and will appear elsewhere \cite{FroLeiTah}.
\end{rem}
The initial conditions in (\ref{eq:eom}) can always be satisfied by going into the initial rest frame of the particle. We will use Kiessling's prescription to determine the force $f$ on the particle. This will depend on the field $u$, which makes (\ref{eq:wave}-\ref{eq:eom}) a coupled system of equations for the joint evolution of the field and its singularity.

Consider first the case of no incoming radiation, i.e. $V_0 \equiv 0$ and $V_1 \equiv 0$. In that case, $u = - \frac{a}{2}|s|$ where $q(t) = 0$ for all $t$ is clearly a time-independent solution to (\ref{eq:wave}), i.e. the particle remains at rest forever. We shall see that in this case, $f \equiv 0$. In this paper we will prove:

%\newpage{}

\begin{theorem}
\hspace{0cm}
\begin{enumerate}[label = (\alph*)]
    \item Suppose $m > 0$. For all smooth initial data $(V_0, V_1)$ for (\ref{eq:wave})  that are compactly supported away from the origin, there exists a solution $(u(t,s),q(t))$ to the joint field-particle evolution problem (\ref{eq:wave}-\ref{eq:eom}), with the property that:
    \begin{enumerate}[label = (\roman*)]
        \item the field $u$ is Lipschitz everywhere and the particle trajectory $q$ is $C^1$,
        \item $u$ is at least $C^1$ away from the particle path $s=q(t)$, and 
        \item for all $\epsilon > 0$, $\exists T>0$ such that $t>T$ implies $|\dot{q}(t)| < \epsilon$.
    \end{enumerate}
    \item Suppose $m < 0$. For all $\epsilon>0$, there exists smooth, compactly supported initial data $(V_0, V_1)$ for (\ref{eq:wave}), with $\|V_0\|_{C^1(\mathbb{R})} + \|V_1\|_{C^0(\mathbb{R})} < \epsilon$, and a solution $(u(t,s),q(t))$ to the joint field-particle evolution problem (\ref{eq:wave}-\ref{eq:eom}), with the property that:
    \begin{enumerate}[label = (\roman*)]
        \item the field $u$ is Lipschitz everywhere and the particle trajectory $q$ is $C^1$,
        \item $u$ is at least $C^1$ away from the particle path $s=q(t)$, and 
        \item the particle reaches the speed of light in finite time, i.e. $\exists T>0$ s.t. $|\dot{q}(T)| = 1$.
    \end{enumerate}
\end{enumerate}
\end{theorem}

Outline of the proof:  In Section 2 we solve the field equations (\ref{eq:wave}) assuming the trajectory $s=q(t)$ of the singularity is given.  We do this by decomposing the field into a smooth part depending only on the incoming radiation, and a singular part sourced by the particle.  In Section 3 we use this field to compute the Kiessling force $f$ in (\ref{eq:eom}), and show that it depends only on the smooth part of the field. We can thus eliminate the field from (\ref{eq:eom}) and have $q(t)$ be the only unknown.  In Section 4 we study (\ref{eq:eom}) by turning it into a dynamical system in the plane and analyzing its phase portrait, which will allow us to prove the stability claim in Section 5 and the instability claim in Section 6. 

We conclude in Section 7 by speculating on the mechanism for instability, and propose various modifications to our model that could perhaps avoid such instabilities.

\section{Solving the field equations}

\begin{proposition}
For any given trajectory $q(t)$ with $|\dot{q}|< 1$, $q(0) = 0$, and $\dot{q}(0) = 0$, the following initial value problem

\begin{equation}
   \begin{cases} 
   \partial_{t}^2 u - \partial_{s}^2 u = a\delta(s - q(t)) \\
   u(0, s) = -\frac{a}{2}\lvert s\rvert + V_0(s)\\
   \partial_{t}u(0, s) = V_1(s) \\
   \end{cases}
\end{equation}
has the unique solution:

\begin{equation}\label{eq:uSol}
    u(t,s) = \frac{a}{2}\begin{cases} 
      s + V(t,s) & s < -t \\
      T_{+}(s + t) - t + V(t,s) & -t < s < q(t) \\
      T_{-}(s - t) - t + V(t,s) & q(t) < s < t \\
      -s + V(t,s) & s > t
    \end{cases}
\end{equation}
where the functions $T_\pm$ are defined by
\begin{equation}
    T_\pm(q(x)\pm x) = x
\end{equation}
for all $x$, i.e. $T_\pm$ is the inverse function to $q(x) \pm x$ (which exists and is $C^1$ so long as $|\dot{q}| < 1$), and 
\begin{equation}
V(t,s) = \frac{1}{2}(V_{0}(s-t) + V_{0}(s + t)) + \frac{1}{2}\int_{s-t}^{s+t} V_{1}(y) dy.
\end{equation}
Furthermore, $u(t,s)$ is at least $C^1$ away from the path $s = q(t)$.
\end{proposition}

\begin{proof}
Define $\Psi(t,s)$ and $\Phi(t,s)$ as solving the following equations: 
\begin{equation}
   \begin{cases} 
      \partial_{t}^2 \Phi - \partial_{s}^2 \Phi = 0\\
      \Phi(0, s) = -\frac{a}{2}\lvert s\rvert + V_0(s)\\
      \partial_{t}\Phi(0, s) = V_1(s) \\
   \end{cases}
%\end{equation}
\qquad
%\begin{equation}
   \begin{cases} 
      \partial_{t}^2 \Psi - \partial_{s}^2 \Psi = a\delta(s - q(t))\\
      \Psi(0, s) = 0\\
      \partial_{t}\Psi(0, s) = 0. \\
   \end{cases}
\end{equation}
Furthermore, define $V(t,s)$ and $U(t,s)$ as solving the following equations:
\begin{equation}
   \begin{cases} 
      \partial_{t}^2 V - \partial_{s}^2 V = 0\\
      V(0, s) = V_0(s)\\
      \partial_{t}V(0, s) = V_1(s) \\
   \end{cases}
%\end{equation}
\qquad
%\begin{equation}
   \begin{cases} 
      \partial_{t}^2 U - \partial_{s}^2 U = 0\\
      U(0, s) = -\frac{a}{2}\lvert s\rvert \\
      \partial_{t}U(0, s) = 0. \\
   \end{cases}
\end{equation}

We thus have that $\Phi = U + V$ and $u = \Psi + \Phi$. Note that because $V_{0}$ and $V_{1}$ are smooth functions, $V(t,s)$ is smooth as well. Hence, $V(t,s)$ contains no singularities.

We can solve for $V$ and $U$ (and hence $\Phi$) by using d'Alembert's formula. We then have the following:

\begin{equation}\label{eq:V}
V(t,s) = \frac{1}{2}(V_{0}(s-t) + V_{0}(s + t)) + \frac{1}{2}\int_{s-t}^{s+t} V_{1}(y) dy,
\end{equation}

\begin{equation}\label{eq:U}
U(t,s) = -\frac{a}{4}(\lvert s-t\rvert + \lvert s + t\rvert) =      \begin{cases} 
      \frac{a}{2}s & s \leq -t \\
      -\frac{a}{2}t & -t < s < t \\
      -\frac{a}{2}s & s \geq t.
    \end{cases}
\end{equation}

We can solve for $\Psi$ with Duhamel's Principle. Define $W(t,s,\tau)$ as follows:
\begin{equation}
\Psi(t,s) = \int_{0}^{t} W(t - \tau, s, \tau) d\tau. 
\end{equation}
It follows that:
\begin{equation}
    \begin{cases} 
        W_{tt} - W_{ss} = 0\\
        W(0,s,\tau) = 0\\
        W_{t}(0,s,\tau) = a\delta(s-q(\tau)).\\
    \end{cases}
\end{equation}

To solve for W, we apply d'Alembert's formula. We have:

\begin{equation}
W(t,s,\tau) = \frac{1}{2}\int_{s-t}^{s+t} a\delta(y-q(\tau)) dy = \frac{a}{2}\chi_{[s-t, s+t]}(q(\tau))
\end{equation}
where $\chi$ is the characteristic function, i.e.:

\begin{equation}
\chi_{[a,b]}(x) = \begin{cases}
1 & a \leq x \leq b\\
0 & \mbox{otherwise}.
\end{cases}
\end{equation}

To integrate $W$ to get $\Psi(t,s)$, we consider the backward light cone of the event $(t,s)$. The $\tau$ for which $(\tau, q(\tau))$ is in this light cone will contribute $\frac{a}{2} d\tau$ to the integral. See Fig.~\ref{fig1}.

Because $q(0) = 0$ and $c = 1$, $q(t)$ is inside the forward light cone drawn from $(0,0)$. As a result of this, $\Psi(t,s) = 0$ when $s > t$ and $s < -t$. Inside the forward light cone of the origin, it is certainly true that the backward light cone of the event $(t,s)$ will intersect with $s = q(t)$. Moreover, it intersects exactly once (going from time 0 to $t$, once $q(t)$ leaves the backward light cone of the event $(t,s)$, it cannot re-enter due to the fact that $c = 1$). We must determine the point at which it intersects, the so-called {\em retarded time}. To the left of $q(t)$, this retarded time $\tau_{2}$ is the solution to $q(\tau_{2}) + \tau_{2} = s + t$, or $T_{+}(s + t)$ for short. The solution is hence $\frac{a}{2}T_{+}(s + t)$. To the right of $q(t)$, this retarded time $\tau_{1}$ is the solution to $q(\tau_{1}) - \tau_{1} = s - t$, or $T_{-}(s - t)$ for short. The solution is hence $\frac{a}{2}T_{-}(s - t)$. 

\begin{figure}[htbp]
\centerline{\includegraphics[scale=.15]{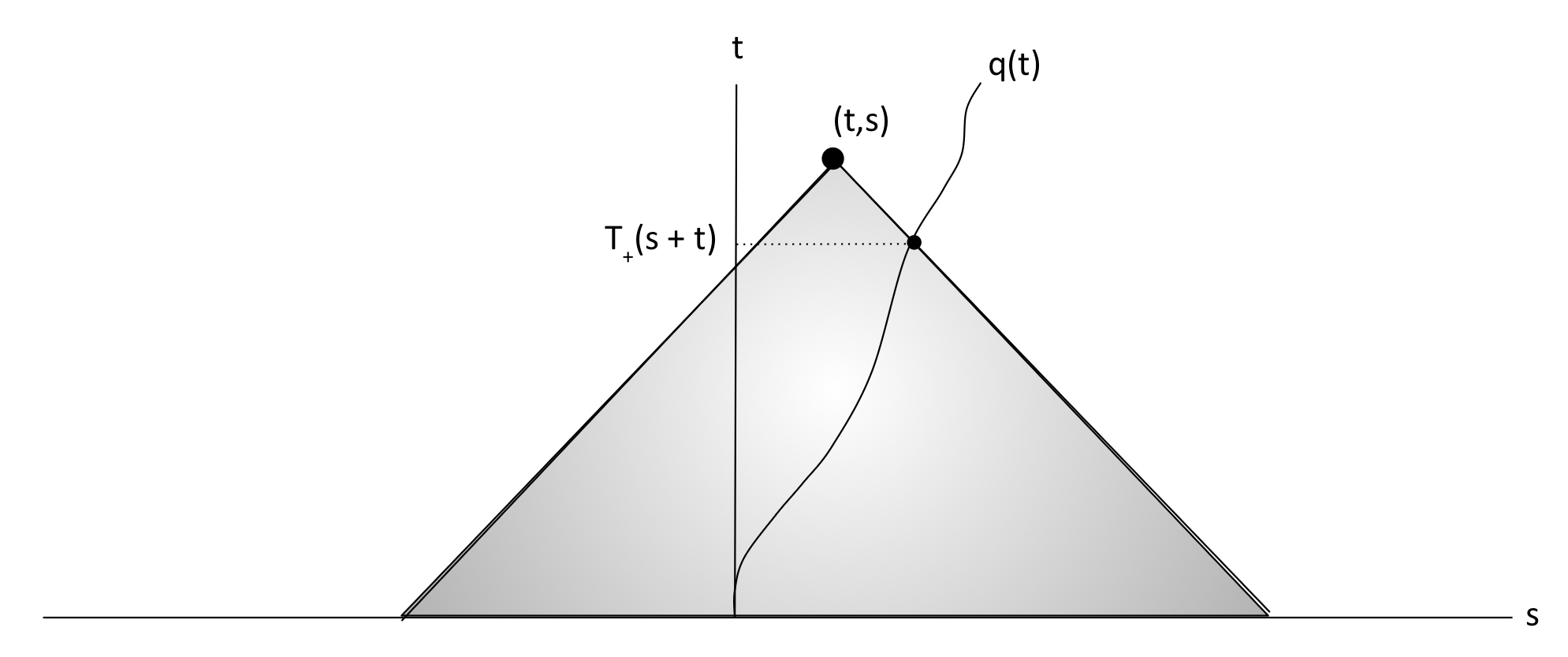}}
\caption{Retarded time $T_{+}$}
\label{fig1}
\end{figure}

%\newpage

We then get the following expression for $\Psi$:

\begin{equation}\label{eq:Psi}
    \Psi(t,s) = \frac{a}{2}\begin{cases} 
      0 & s < -t \\
      T_{+}(s + t) & -t < s < q(t) \\
      T_{-}(s - t) & q(t) < s < t \\
      0 & s > t.
    \end{cases}
\end{equation}
We see from here that $\Psi$, like $U$, is $C^0$ but not $C^1$ because it has two singularities at $s + t = 0$ and $s - t = 0$, i.e. along the light cone of the origin. We will see in Proposition 2 that when we add $\Psi$ and $U$, the singularities along the light cone cancel each other out.

The full solution $u(t,s)$ is thus as follows:

\begin{equation}\label{eq:u}
    u(t,s) = \frac{a}{2}\begin{cases} 
      s + V(t,s) & s < -t \\
      T_{+}(s + t) - t + V(t,s) & -t < s < q(t) \\
      T_{-}(s - t) - t + V(t,s) & q(t) < s < t \\
      -s + V(t,s) & s > t
    \end{cases}
\end{equation}
where $V$ is as defined in Equation (\ref{eq:V}).
\end{proof}

We expect $u(t,s)$ to have singularities only at $s = q(t)$. We pause for a moment to show that $u$ has no singularities along the lines $s = -t$ and $s = t$. 

\hspace{0cm}

\begin{proposition}
For $u(t,s)$ given by Equation (\ref{eq:uSol}), $u(t,s)$ is $C^{1}$ across $s = -t$ and $s = t$. 

\end{proposition}

\begin{proof}
Because $V_{0}$ and $V_{1}$ are smooth, $V(t,s)$ is smooth. Hence, it suffices to look at the singular part of $u(t,s)$. Let:
\begin{equation}
      w(t,s) = U(t,s) + \Psi(t,s) = \frac{a}{2}\begin{cases} 
      s & s < -t \\
      T_{+}(s + t) - t & -t < s < q(t) \\
      T_{-}(s - t) - t & q(t) < s < t \\
      -s & s > t.
    \end{cases}
\end{equation}

We will first show $w(t,s)$ is $C^{1}$ across $s = -t$. We have:
\begin{equation}\label{leftlim}
\partial_{s}w|_{(s = -t)^{-}} = \partial_{s}(\frac{a}{2}s) = \frac{a}{2},
%\end{equation}
\quad
%\begin{equation}
\partial_{t}w|_{(s = -t)^{-}} = \partial_{t}(\frac{a}{2}s) = 0.
\end{equation}
Recall that $T_{+}(s + t) = \tau_{2}$ where $\tau_{2}$ solves the following:
\begin{equation}
q(\tau_{2}) + \tau_{2} = s + t.
\end{equation}
Using implicit differentiation by $s$ and $t$ yields:

\begin{equation}
\partial_{s}\tau_{2}\  \dot{q}(\tau_{2}) + \partial_{s}\tau_{2} = 1,
\quad
\partial_{t}\tau_{2}\   \dot{q}(\tau_{2}) + \partial_{t}\tau_{2} = 1
\end{equation}
respectively. At the line $s = -t$, $\tau_{2} = 0$, so $\dot{q}(\tau_{2}) = \dot{q}(0) = 0$. Hence:
\begin{equation}
\partial_{s}\tau_{2} = \partial_{t}\tau_{2} = 1.
\end{equation}
We thus have:
\begin{equation}
\partial_{s}w|_{(s = -t)^{+}} =  \frac{a}{2}\partial_{s}T_{+}(s + t) = \frac{a}{2}\partial_{s}\tau_{2} = \frac{a}{2}
\end{equation}
and
\begin{equation}
\partial_{t}w|_{(s = -t)^{+}} = \frac{a}{2}(\partial_{t}T_{+}(s + t) - 1) = \frac{a}{2}(\partial_{s}\tau_{2} - 1) = 0.
\end{equation}
Comparing with (\ref{leftlim}) shows that $w$ is $C^1$ across $s=-t$.  

The proof that
 $w(t,s)$ is $C^{1}$ across $s = t$ is completely analogous.

\end{proof}

\section{Computing the Kiessling force}
We would like to combine the solution \eqref{eq:uSol} for $u(t,s)$ with equation \eqref{eq:eom} to find a system of ODEs for $q(t)$. To do this, we need to work out the Kiessling force $f$ in \eqref{eq:eom}.

We begin by recalling that Kiessling {\em postulates} the local conservation of total energy-momentum for the field and particle system:
\begin{equation}\label{eq:cons}
\partial^{\mu}T^{\mbox{\tiny total}}_{\mu\nu} = 0.
\end{equation}
Here, $T^{\mbox{\tiny total}}_{\mu\nu}$ is the energy density-momentum density-stress tensor (or {\em energy-momentum tensor}, for short) for the field and particle system. To find the energy-momentum tensor for the field, we start with the Lagrangian. The Lagrangian for the massless scalar field is: 

\begin{equation}
\mathcal{L} = \frac{1}{2}\eta^{\mu\nu}\partial_{\mu}u\partial_{\nu}u.
\end{equation}
Here, $\eta = \mbox{diag}(1,-1)$ is the Minkowski metric.
The energy-momentum tensor for the field is defined as:

\begin{equation}
T_{\mu\nu}^{\mbox{\tiny field}} = 2\frac{\partial{\mathcal{L}}}{\partial{\eta^{\mu\nu}}} - \eta_{\mu\nu}\mathcal{L},
\end{equation}
and thus in this case
\begin{equation}\label{def:Tfield}
T_{\mu\nu}^{\mbox{\tiny field}} = \partial_\mu u\partial_\nu u - \frac{1}{2}\eta_{\mu\nu}\partial_\alpha u\partial^\alpha u.
\end{equation}
Since $u$ is expected to be singular on the worldline of the particle, the above is only well-defined away from the particle path, but our assumptions on the field are such that $T^{\mbox{\tiny field}}$ can be continued into the particle path as a spacetime distribution.

The energy-momentum tensor for the particle on the other hand is defined as a distribution on spacetime that is concentrated on the worldline $x^\mu = z^\mu(\tau)$ of the particle ($\tau$ is the arclength parameter):
\begin{equation}
    T^{\mbox{\tiny particle}}_{\mu\nu} := m \int \mathbf{u}_\mu \mathbf{u}_\nu \delta^{(2)}(x-z(\tau)) d\tau = \frac{m}{\mathbf{u}^0}\mathbf{u}_\mu \mathbf{u}_\nu \delta(s-q(t)),
\end{equation}
where $\mathbf{u}$ is the unit tangent to the worldline of the particle
\begin{equation}
   \mathbf{u}^\mu := dz^\mu/d\tau,\qquad \mathbf{u}_\mu \mathbf{u}^\mu = 1.
\end{equation}
The definition of $T^{\mbox{\tiny particle}}$ is such that:
\begin{equation}
\partial^{\mu}T_{\mu\nu}^{\mbox{\tiny particle}} = \mathbf{f}_{\nu}(t)\delta(s-q(t))
\end{equation}
holds, where the spacetime covector $\mathbf{f}_\nu$ is the 2-force acting on the particle (cf. \cite{kiessling}, eq. 72).

Let us take a second to consider how this relates to the $f(t,q,\dot{q})$, the force on the particle, that appears in \eqref{eq:eom}. There, $f(t,q,\dot{q})$ is clearly the spatial component of a spacetime (contravariant) vector. We therefore must have
\begin{equation}
    f(t,q(t),\dot{q}(t)) = \mathbf{f}^1(t) = - \mathbf{f}_1(t)
\end{equation}
since we have chosen the signature $(+,-)$ for the Minkowski metric.

Hence, setting $\nu = 1$:
\begin{equation}
\partial^{\mu}T^{\mbox{\tiny particle}}_{\mu 1}(t,s) = -f(t,q(t),\dot{q}(t))\delta(s - q(t)).
\end{equation}

Going back to the energy-momentum tensor for the field, we have for $\nu = 1$:

\begin{equation}
\partial^{\mu}T^{\mbox{\tiny field}}_{\mu 1} = \partial^{0}T^{\mbox{\tiny field}}_{01} + \partial^{1}T^{\mbox{\tiny field}}_{11} = \partial_{t}\pi - \partial_{s}\tau
\end{equation}
Using equation (\ref{eq:cons}), we have:

\begin{equation}
0=\partial^{\mu}T^{\mbox{\tiny total}}_{\mu 1} = \partial^{\mu}T_{\mu 1}^{\mbox{\tiny field}} + \partial^{\mu}T_{\mu 1}^{\mbox{\tiny particle}} = \partial_{t}\pi - \partial_{s}\tau - f(t,q,\dot{q})\delta(s - q(t)).
\end{equation}

Rearranging this gives us the momentum-balance law:

\begin{equation}\label{eq:mombal}
\partial_{t}\pi - \partial_{s}\tau = f(t,q,\dot{q})\delta(s - q(t)).
\end{equation}
From \eqref{def:Tfield} we have
\begin{equation}
\pi(t,s) = T^{\mbox{\tiny field}}_{01} =  u_{s}u_{t},
\end{equation}
\begin{equation}
\tau(t,s) = T^{\mbox{\tiny field}}_{11} = \frac{1}{2}(u_{s}^2+u_{t}^2).
\end{equation}

\begin{proposition}\label{prop:kf}
Assume that the field $u$ is Lipschitz continuous in a tubular neighborhood of the $C^1$ path $(t,q(t))$ of the particle, and that $u$ is $C^1$ on either side of the path.  Then the force appearing in (\ref{eq:mombal}) is 
\begin{equation}\label{def:kforce}
   f(t,q(t),\dot{q}(t)) = -\dot{q}[\pi(t,s)]_{s = q(t)} - [\tau(t,s)]_{s = q(t)} 
\end{equation}
where $[ \cdot ]_{s=q(t)}$ denotes the jump across the path.
\end{proposition}

\begin{proof}

\begin{figure}
\centering
\includegraphics[scale=.17]{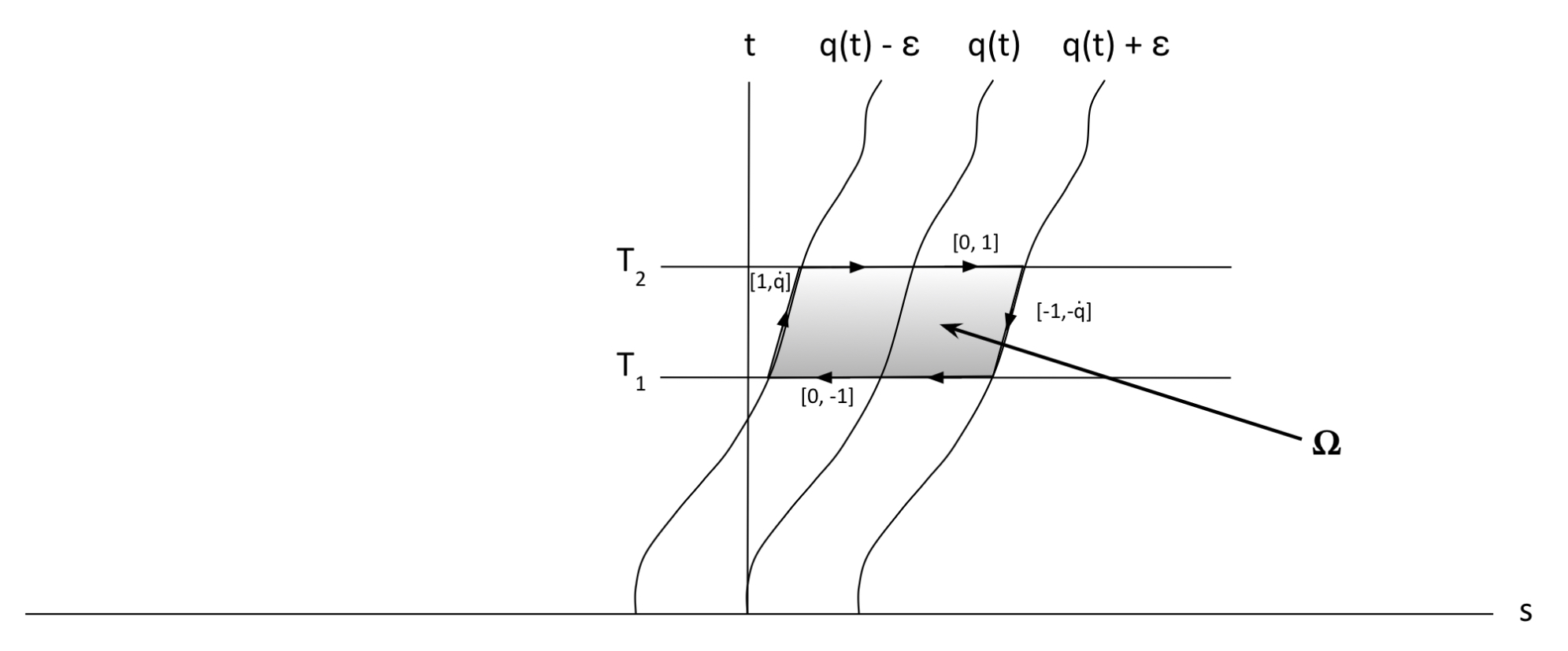}
\caption{Region of integration for the momentum-balance equation}
\label{fig2}
\end{figure}

Note that the assumptions imply that $\pi$ and $\tau$ are bounded and can at most have a jump discontinuity across the path, so that the right-hand side of (\ref{def:kforce}) is well-defined.
Let $\epsilon > 0$ and $T_{2} > T_{1} \geq 0$. We will integrate \eqref{eq:mombal} over the region  $\Omega = \{(t, s) | T_1 \leq t \leq T_2, q(t) - \epsilon \leq s \leq q(t) + \epsilon\}$ and then take the limit as $\epsilon$ goes to $0$. After integrating and taking the limit, the right-hand side of equation \eqref{eq:mombal} becomes:
\begin{equation}
\int_{T_{1}}^{T_{2}} f(t,q(t),\dot{q}(t)) dt.
\end{equation}
After integrating and using Green's Theorem, the left-hand side of \eqref{eq:mombal} becomes:
\begin{eqnarray*}
\int_{q(T_2) - \epsilon}^{q(T_2) + \epsilon} \pi(T_2, s)ds  &- & \int_{q(T_1) - \epsilon}^{q(T_1) + \epsilon} \pi(T_1, s)ds \\
& - & \left.\int_{T_{1}}^{T_{2}} \dot{q}\pi(t,s) + \tau(t,s)\right|_{s = q(t) + \epsilon} dt \\
& + & \left.\int_{T_{1}}^{T_{2}} \dot{q}\pi(t,s) + \tau(t,s)]\right|_{s = q(t) - \epsilon} dt.
\end{eqnarray*}
Because $\pi$ is locally integrable, the first two terms go to $0$ as $\epsilon \to 0$. Taking the limit as $\epsilon$ goes to $0$ in the other two terms gives us:
\begin{equation}
- \int_{T_{1}}^{T_{2}} [\dot{q}\pi(t,s) + \tau(t,s)]_{s = q(t)} dt
\end{equation}
for the left-hand side. Because $T_{1}$ and $T_{2}$ were arbitrary, we therefore have:
\begin{equation}
f(t,q(t),\dot{q}(t)) = -\dot{q}[\pi(t,s)]_{s = q(t)} - [\tau(t,s)]_{s = q(t)}.
\end{equation}
\end{proof}
\begin{proposition}
Assume $u(t,s)$ is a solution to the joint evolution problem (\ref{eq:wave}-\ref{eq:eom}). The Kiessling force in (\ref{eq:eom}) is given by:
\begin{equation}
f(t,q,\dot{q}) = aV_s(t,q) - \frac{a^2}{2}\frac{\dot{q}}{1-\dot{q}^2}
\end{equation}
\end{proposition}
\begin{proof}
By Prop.~\ref{prop:kf}, the force is:
\begin{equation}\label{eq:kforcejump}
f(t,q(t),\dot{q}(t)) = -\dot{q}[\pi]_{s = q(t)} - [\tau]_{s = q(t)} = -\dot{q}[u_{s}u_{t}]_{s = q(t)} - \frac{1}{2}[u_{s}^2 + u_{t}^2]_{s = q(t)}
\end{equation}
using $u(t,s)$ as given by equation (\ref{eq:uSol}). Substituting in $u = V + w$ gives us

\begin{equation}
[u_tu_s]_{s=q(t)} = V_t[w_s]_{s = q(t)} + V_s[w_t]_{s = q(t)} + [w_sw_t]_{s = q(t)}
\end{equation}
and
\begin{equation}
[\frac{1}{2}u^2_t + \frac{1}{2}u^2_s]_{s = q(t)} = \frac{1}{2}[w^2_t]_{s = q(t)} + \frac{1}{2}[w^2_s]_{s = q(t)} + V_s[w_s]_{s = q(t)} + V_t[w_t]_{s = q(t)}
\end{equation}
where we've used the fact that $V$ is smooth. To determine the necessary values, we will first compute $w_{s}|_{s=q(t)^{-}}$, $w_{s}|_{s=q(t)^{+}}$, $w_{t}|_{s=q(t)^{-}}$, and $w_{t}|_{s=q(t)^{+}}$. We have:

\begin{equation}
w_{s}|_{s=q(t)^{-}} = \partial_{s}(\frac{a}{2}(T_{+}(s+t) - t)) = \frac{a}{2}\frac{1}{\dot{q}(t) + 1},
\end{equation}

\begin{equation}
w_{s}|_{s=q(t)^{+}} = \partial_{s}(\frac{a}{2}(T_{-}(s-t) - t)) = \frac{a}{2}\frac{1}{\dot{q}(t) - 1},
\end{equation}

\begin{equation}
w_{t}|_{s=q(t)^{-}} = \partial_{t}(\frac{a}{2}(T_{+}(s+t) - t)) = -\frac{a}{2}\frac{\dot{q}(t)}{\dot{q}(t) + 1},
\end{equation}

\begin{equation}
w_{t}|_{s=q(t)^{+}} = \partial_{t}(\frac{a}{2}(T_{-}(s-t) - t)) = -\frac{a}{2}\frac{\dot{q}(t)}{\dot{q}(t) - 1}.
\end{equation}

Thus, our final results for $[w_{s}]_{s=q(t)}$, $[w_{t}]_{s=q(t)}$, $[w_{s}w_{t}]_{s=q(t)}$, $[w_{s}^2]_{s = q(t)}$, and $[w_{t}^2]_{s = q(t)}$ are:

\begin{equation}
[w_{s}]_{s=q(t)} = \frac{a}{2}(\frac{1}{\dot{q}(t) - 1} - \frac{1}{\dot{q}(t) + 1}) = \frac{a}{\dot{q}(t)^{2} - 1},
\end{equation}

\begin{equation}
[w_{t}]_{s=q(t)} = -\frac{a\dot{q}(t)}{2}(\frac{1}{\dot{q}(t) - 1} - \frac{1}{\dot{q}(t) + 1}) = -\frac{a\dot{q}(t)}{\dot{q}(t)^{2} - 1},
\end{equation}

\begin{equation}
[w_{s}w_{t}]_{s=q(t)} = -\frac{a^{2}\dot{q}(t)}{4}(\frac{1}{(\dot{q}(t) - 1)^{2}} - \frac{1}{(\dot{q}(t) + 1)^{2}}) = -\frac{a^{2}\dot{q}(t)^{2}}{(\dot{q}(t)^{2} - 1)^{2}},
\end{equation}

\begin{equation}
[w^2_t]_{s=q(t)} = \frac{a^2}{4}((\frac{\dot{q}(t)}{\dot{q}(t) - 1})^2 - (\frac{\dot{q}(t)}{\dot{q}(t) + 1})^2) = \frac{a^2\dot{q}(t)^3}{(\dot{q}(t)^2 - 1)^2},
\end{equation}

\begin{equation}
[w^2_s]_{s=q(t)} = \frac{a^2}{4}((\frac{1}{\dot{q}(t) - 1})^2 - (\frac{1}{\dot{q}(t) + 1})^2) = \frac{a^2\dot{q}(t)}{(\dot{q}(t)^2 - 1)^2}.
\end{equation}

Inserting these values into equation (\ref{eq:kforcejump}) gives us

\begin{equation}
f(t,q(t),\dot{q}(t)) = aV_s(t,q(t)) - \frac{a^2}{2}\frac{\dot{q}(t)}{1-\dot{q}(t)^2}.
\end{equation}

\end{proof}

Note that the first term represents the force that the external field is exerting on the particle. That is, the first term is usually taken to be the force acting on a scalar particle. The second term represents the {\em self-force} (the force the particle exerts on itself), which here is in the opposite direction of the motion.

\section{Equations of motion as a dynamical system}
We can now look at the equations of motion for the particle, which are the following:

\begin{equation}
 \begin{cases} 
      \dot{q} = \frac{p}{m\sqrt{1 + \frac{p^2}{m^2}}} \\
      \dot{p} = aV_s(t,q(t)) - \frac{a^2}{2}\frac{\dot{q}(t)}{1-\dot{q}(t)^2}. \\
   \end{cases}
\end{equation}

We substitute the expression for $\dot{q}$ into the equation of $\dot{p}$, which results in the following equation: 
\begin{equation}
\dot{p} = aV_s(t,q(t)) - \frac{a^2}{2}\frac{p}{m} \sqrt{1 + \frac{p^2}{m^2}}.
\end{equation}
In addition to this, let us rewrite $V_{s}(t,q(t))$. Recall that:
\begin{equation}
V(t,s) = \frac{1}{2}(V_{0}(s+t) + V_{0}(s-t)) + \frac{1}{2}\int_{s-t}^{s+t} V_{1}(x) dx.
\end{equation}
Hence, we have that:
\begin{equation}
V_{s}(t,s) = \frac{1}{2}(\dot{V_{0}}(s+t) + \dot{V_{0}}(s-t) + V_{1}(s+t) - V_{1}(s-t)).
\end{equation}
Let us further define $F(s)$ and $G(s)$ as follows:
\begin{equation}\label{eq:fandg}
 \begin{cases} 
      F(s) = \dot{V_{0}}(s) + V_{1}(s) \\
      G(s) = \dot{V_{0}}(s) - V_{1}(s). \\
   \end{cases}
\end{equation}
From our definitions of $F$ and $G$, we can rewrite our equations of motion, specifically the expression for $\dot{p}$. It now becomes:
\begin{equation}
 \begin{cases} 
      \dot{q} = \frac{p}{m\sqrt{1 + \frac{p^2}{m^2}}}\\
      \dot{p} = \frac{a}{2}(F(q + t) + G(q - t)) - \frac{a^2}{2}\frac{p}{m} \sqrt{1 + \frac{p^2}{m^2}}. \\
   \end{cases}
\end{equation}

We can further simplify our equations using a change of variables so that we can get rid of the square roots. We will let $\frac{p}{m} = \tan\theta$, so our new equations become:

\begin{equation}
 \begin{cases} 
      \dot{q} = \sin\theta \\
      \dot{\theta} = \frac{a}{2m}(F(q + t) + G(q - t))\cos^2\theta - \frac{a^2}{2m}\sin\theta. \\
   \end{cases}
\end{equation}

To get an autonomous system, we define new unknowns:

\begin{equation}
 \begin{cases} 
      d(t) = q(t) + t \\
      b(t) = q(t) - t \\
   \end{cases}
\end{equation}
and write the system of three equations as follows:
\begin{equation}\label{eq:sys3}
 \begin{cases} 
      \dot{d} = \sin\theta + 1 \\
      \dot{b} = \sin\theta - 1 \\
      \dot{\theta} = \frac{a}{2m}(F(d) + G(b))\cos^2{\theta} - \frac{a^2}{2m}\sin\theta. \\
   \end{cases}
\end{equation}

To solve this, we need to look for a solution $(b,d, \theta)$ such that:
\begin{equation}\label{eq:sys3init}
 \begin{cases} 
      b(0) = 0\\
      d(0) = 0\\
      \theta(0) = 0.\\
   \end{cases}
\end{equation}
These are the consequences of our initial conditions $q(0) = 0$ and $\dot{q}(0) = 0$. Furthermore, we know the following limits for each variable: $0 < d(t) < 2t$, $-2t < b(t) < 0$, $-\pi/2 \leq \theta(t) \leq \pi/2$.

We will split this into two cases: one where the bare mass, $m$, is positive and one where it is negative. For the first, we will prove that the solution will always be stable. For the second, we will show a case where the solution is unstable. 

\section{Proof of stability for positive bare mass}
In the case of positive bare mass, we are concerned with (\ref{eq:sys3}) with $m > 0$. Recall that because $F(d)$ and $G(b)$ are defined as in (\ref{eq:fandg}), they must be compactly supported. Furthermore, note that it must always be the case that $\dot{d} = \sin\theta + 1 \geq 0$ and that $\dot{b} = \sin\theta - 1 \leq 0$. Hence, it suffices to look only at the region where $b \leq 0 $ and $d \geq 0$. We will define $[b_{L},b_{R}]$ such that $-\infty < b_{L} < b_{R} \leq 0$ and $G(b) = 0$ for $b$ outside $[b_{L},b_{R}]$. Similarly, we will define $[d_{L},d_{R}]$ such that $0 \leq d_{L} < d_{R} < \infty$ and $F(d) = 0$ for $d$ outside $[d_{L},d_{R}]$. Based on the fact that $\dot{d} \geq 0$ and $\dot{b} \leq 0$, Figure \ref{fig3} shows a rough sketch of the trajectory of the solution projected onto the $(b,d)$ plane. 

\begin{figure}[htbp]
\centerline{\includegraphics[scale=.16]{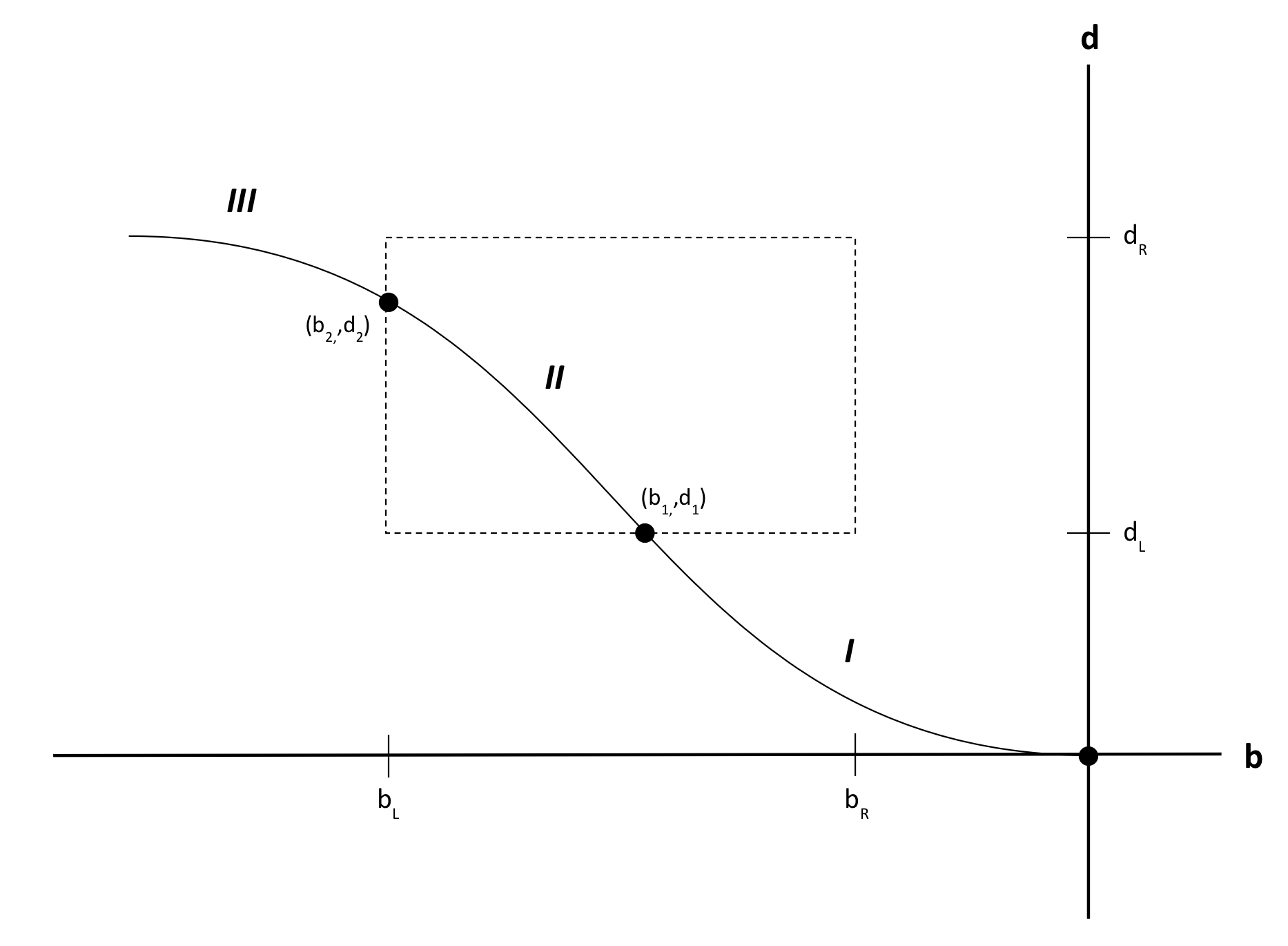}}
\caption{Projection of trajectory of the solution to (\ref{eq:sys3}) in the (b,d) plane}
\label{fig3}
\end{figure}

Based on this, we will divide this analysis into three regions: before the radiation, during the radiation, and after the radiation. In the first region, $G(b) = F(d) = 0$, so (\ref{eq:sys3}) reduces to:

\begin{equation}\label{eq:sys3norad}
 \begin{cases} 
      \dot{d} = \sin\theta + 1 \\
      \dot{b} = \sin\theta - 1 \\
      \dot{\theta} = -\frac{a^2}{2m}\sin\theta. \\
   \end{cases}
\end{equation}

One can check that the unique solution to (\ref{eq:sys3norad}) with initial conditions in (\ref{eq:sys3init}) is:

\begin{equation}
 \begin{cases} 
      d(t) = t \\
      b(t) = -t \\
      \theta(t) = 0.\\
   \end{cases}
\end{equation}

Hence, the particle will have the following conditions entering the second region:
\begin{equation}\label{eq:cond2}
 \begin{cases} 
      b(t_{1}) = b_{1}\\
      d(t_{1}) = d_{1}\\
      \theta(t_{1}) = 0\\
   \end{cases}
\end{equation}
where $t_{1} > 0$. In the second region, note that when $\theta = \frac{\pi}{2}$, $\dot{\theta} = -\frac{a^2}{2m} < 0$. Similarly, when $\theta = -\frac{\pi}{2}$, $\dot{\theta} = \frac{a^2}{2m} > 0$. Hence, the trajectory can never cross $\theta = \frac{\pi}{2}$ and $\theta = -\frac{\pi}{2}$. The particle will thus end up with the following conditions entering the third region:

\begin{equation}\label{eq:cond3}
 \begin{cases} 
      b(t_{2}) = b_{2}\\
      d(t_{2}) = d_{2}\\
      \theta(t_{2}) = \theta_{2}\\
   \end{cases}
\end{equation}
where $t_{2} > 0$ and $-\frac{\pi}{2} < \theta_{2} < \frac{\pi}{2}$. Hence, the third interval will amount to solving the following set of ODEs with conditions specified in (\ref{eq:cond3}):

\begin{equation}\label{eq:sys3noradii}
 \begin{cases} 
      \dot{b} = \sin\theta - 1 \\
      \dot{d} = \sin\theta + 1 \\
      \dot{\theta} = -\frac{a^2}{2m}\sin\theta. \\
   \end{cases}
\end{equation}

We solve the third ODE explicitly in Proposition 5. 

\hspace{0cm}

\begin{proposition}
Suppose we have (\ref{eq:sys3noradii}) with initial conditions in (\ref{eq:cond3}).

\begin{enumerate}[label = (\alph*)]
\item If $\theta_{2} = 0$, $\theta(t) = 0$ for $t > t_{2}$.
\item If $\frac{\pi}{2} > \theta_{2} > -\frac{\pi}{2}$ and $\theta_{2} \neq 0$, then $\underset{t \rightarrow \infty}{lim} \theta(t) = 0$.
\end{enumerate}
\end{proposition}

\begin{proof}
$\theta(t) = 0$ is a trivial solution which satisfies $\theta(t_{2}) = 0$. Noting that $-\frac{a^2}{2m}\sin\theta$ and its derivative with respect to $\theta$ is continuous everywhere, we see that such a solution is unique. Hence, (a) follows.

For (b), assume $\theta_{2} > 0$. The proof is similar for $\theta_{2} < 0$. If $\theta = 0$ for some time $t_{3} > t_{2}$, we are left with (a), and $\theta(t) = 0 < \epsilon$ for $t > t_{3}$. Hence, assume $\theta > 0$ at all times. Then $\sin\theta \neq 0$, and we can separate the third equation of (\ref{eq:sys3noradii}):
\begin{equation}
\frac{1}{\sin\theta} d\theta = -\frac{a^2}{2m} dt. \\
\end{equation}
Integrating, we have:
\begin{equation}
-\ln|\csc(\theta) + \cot(\theta)| = -\frac{a^2}{2m}t + C_{0} \\
\end{equation}
or
\begin{equation}
\csc(\theta) + \cot(\theta) = C_{1}e^{\frac{a^2}{2m}t}.
\end{equation}

We get rid of the absolute value by choosing the sign for $C_{1}$. Here, $C_{1} > 0$ since $\theta(t_{2}) > 0$ implies $\csc(\theta) + \cot(\theta) > 0$. For any $\epsilon > 0$, we can choose $t_{3}$ such that $\csc(\epsilon) + \cot(\epsilon) = C_{1}e^{\frac{a^2}{2m}t_{3}}$. Then for $t > t_{3}$, $\csc(\theta(t)) + \cot(\theta(t)) = C_{1}e^{\frac{a^2}{2m}t} > C_{1}e^{\frac{a^2}{2m}t_{3}} = \csc(\epsilon) + \cot(\epsilon)$. Thus, $\theta(t) < \epsilon$.
\end{proof}

With this, we have proved the stability of the solution for positive bare mass. In the case where $\theta_{2} = 0$, we see that $\theta(t) = 0$ for $t > t_{2}$. Recalling that $\dot{q} = \sin\theta$, we have $\dot{q}(t) = 0$ for $t > t_{2}$. In the case where  $\theta_{2} \neq 0$, we see that since $\underset{t \rightarrow \infty}{\text{lim}} \theta(t) = 0$ and $\dot{q} = \sin\theta$, $\underset{t \rightarrow \infty}{\text{lim}} \dot{q}(t) = 0$. In either case, the speed of the particle goes to zero.

\section{Proof of instability for negative bare mass}

In the case of negative bare mass, we are concerned with (\ref{eq:sys3}) with $m < 0$. We will show that a specific case of radiation leads to an unstable solution. Assume that the radiation is purely incoming: $F \equiv 0$. Set $G_{\beta}(x) = \beta\sin(\pi x)\chi_{[-3,-1]}$ where $\beta \in \mathbb{R}$. We can ignore the $d$ equation and are left with a system of two equations:

\begin{equation}\label{eq:sys2pre}
 \begin{cases} 
      \dot{b} = \sin\theta - 1 \\
      \dot{\theta} = \frac{a \cos^2\theta}{2m}(G_{\beta}(b)) - \frac{a^2}{2m}\sin\theta, \\
   \end{cases}
\end{equation}

\begin{equation}
 \begin{cases} 
      b(0) = 0\\
      \theta(0) = 0.\\
   \end{cases}
\end{equation}

We will rewrite (\ref{eq:sys2pre}) by replacing $m$ with $-|m|$:

\begin{equation}\label{eq:sys2}
 \begin{cases} 
      \dot{b} = \sin\theta - 1 \\
      \dot{\theta} = - \frac{a }{2|m|}G_{\beta}(b)\cos^2\theta + \frac{a^2}{2|m|}\sin\theta, \\
   \end{cases}
\end{equation}

\begin{equation}\label{eq:sys2init}
 \begin{cases} 
      b(0) = 0\\
      \theta(0) = 0.\\
   \end{cases}
\end{equation}

In addition, it will be useful to work with the reverse flow of the system:

\begin{equation}\label{eq:rev}
 \begin{cases} 
      \dot{b} = -\sin\theta + 1 \\
      \dot{\theta} = \frac{a \cos^2\theta}{2|m|}(G_{\beta}(b)) - \frac{a^2}{2|m|}\sin\theta. \\
   \end{cases}
\end{equation}

Note that $\beta = 0$ corresponds to the static solution. To see this, note that (\ref{eq:sys2}) becomes:
\begin{equation}\label{eq:sys2norad}
 \begin{cases} 
      \dot{b} = \sin\theta - 1 \\
      \dot{\theta} = \frac{a^2}{2|m|}\sin\theta. \\
   \end{cases}
\end{equation}
The solution to this with initial conditions given in (\ref{eq:sys2init}) is:
\begin{equation}\label{eq:static}
 \begin{cases} 
      b(t) = -t \\
      \theta(t) = 0. \\
   \end{cases}
\end{equation}
Since $\theta = 0$ corresponds to $\dot{q} = 0$, (\ref{eq:static}) is the static solution. 

To get a better sense of the system of ODEs in (\ref{eq:sys2}) for non-zero $\beta$, see Figure \ref{fig4}. The interval $[-3,-1]$ represents the particle becoming perturbed by the incoming radiation. A particle with initial conditions given in (\ref{eq:sys2init}) will end up with $b = -1$ and $\theta = 0$. To see this, note that outside of $[-3, -1]$, $G_{\beta}(b) = 0$, and (\ref{eq:sys2}) reduces to (\ref{eq:sys2norad}). Hence, in the interval $[-1, 0]$, the solution to (\ref{eq:sys2}) with initial conditions given in (\ref{eq:sys2init}) is (\ref{eq:static}). At $t_{1} = 1$, $b(t_{1}) = -1$ and $\theta(t_{1}) = 0$. At this point, we can take the system of ODEs in (\ref{eq:sys2init}) and modify the initial conditions as follows:

\begin{equation}
 \begin{cases} 
      b(t_{1}) = -1\\
      \theta(t_{1}) = 0.\\
   \end{cases}
\end{equation}

After entering the interval $[-3, -1]$, Figure \ref{fig4} suggests that the particle will oscillate. The question is whether or not the particle will go back to rest ($\theta = 0$ at $b = -3$, represented by the grey line in Figure \ref{fig4}) or be left with some speed ($\theta \neq 0$ at $b = -3$). In the former case, the particle will remain at rest. In the latter case, the particle will go toward $\theta = \pm\frac{\pi}{2}$. Recalling that $\dot{q} = \sin\theta$, this means $\dot{q} = \pm 1$. That is, the particle will reach the speed of light in finite time. This is proved formally in the next proposition.

\hspace{0cm}
\newpage{}

\begin{proposition}
Suppose for some $t_{0} > 0$, $b(t_{0}) = -3$.

\begin{enumerate}[label = (\alph*)]
\item If $\theta_{0} = \theta(t_{0}) = 0$, $\theta(t) = 0$ for $t > t_{0}$.
\item If $\frac{\pi}{2} > \theta(t_{0}) > 0$, $\theta(t_{1}) = \frac{\pi}{2}$ for some $t_{1} > t_{0}$.
\item If $-\frac{\pi}{2} < \theta(t_{0}) < 0$, $\theta(t_{1}) = -\frac{\pi}{2}$ for some $t_{1} > t_{0}$.
\end{enumerate}
\end{proposition}

\begin{proof}
Because $\dot{b} =$ $\sin\theta - 1$, it is always true that $\dot{b} \leq 0$. Hence, for $t > t_{0}$, $b < -3$, and $G(b) = 0$. We can then rewrite the ODEs in (\ref{eq:sys2}) as (\ref{eq:sys2norad}). $\theta(0) = 0$ is a trivial solution and satisfies $\theta(t_{0}) = 0$. Since such a solution is unique, (a) follows.

We will now show (b). Because $\theta(t_{0}) > 0$ and $\dot{\theta} > 0$ if $\theta > 0$, it follows that $\theta(t_{1}) > 0$ for $t_{1} > t_{0}$. Hence, $\sin\theta \neq 0$, and we can separate the second equation of (\ref{eq:sys2norad}):
\begin{equation}
\frac{1}{\sin\theta} d\theta = \frac{a^2}{2|m|} dt. \\
\end{equation}
Integrating, we have:
\begin{equation}
-\ln|\csc(\theta) + \cot(\theta)| = \frac{a^2}{2|m|}t + C_{0} \\
\end{equation}
or
\begin{equation}
\csc(\theta) + \cot(\theta) = C_{1}e^{-\frac{a^2}{2|m|}t}.
\end{equation}

We get rid of the absolute value by choosing the sign for $C_{1}$. Here, $C_{1} > 0$ since $\theta(t_{0}) > 0$ implies $\csc(\theta) + \cot(\theta) > 0$. Now, for $0 \leq \theta \leq \pi$, $0 \leq \csc(\theta) + \cot(\theta) < \infty$, is monotonously decreasing, and is invertible. Let $\Theta: [0, \pi] \longrightarrow [0, \infty)$ be the inverse of $\csc(\theta) + \cot(\theta)$. We can now write the solution explicitly as: 

\begin{equation}
\theta(t) = \Theta(C_{1}e^{-\frac{a^2}{2|m|}t}).
\end{equation}

The initial conditions tell us:
\begin{equation}
C_{1} = \frac{\csc(\theta(t_{0})) + \cot(\theta(t_{0}))}{e^{-\frac{a^2}{2|m|}t_{0}}} = \frac{C_{2}}{e^{-\frac{a^2}{2|m|}t_{0}}}.
\end{equation}
where $C_{2} = \csc(\theta(t_{0})) + \cot(\theta(t_{0}))$. Because $0 \leq \theta(t_{0}) < \frac{\pi}{2}$, $C_{2} > 1$. Substituting $C_{1}$ gives us: 
\begin{equation}
\theta(t) = \Theta(C_{2}e^{-\frac{a^2}{2|m|}(t - t_{0})}).
\end{equation}
Consider $t_{1} = -\frac{2|m|}{a^2}\ln(\frac{1}{C_{2}}) + t_{0} > t_{0}$. Using the fact that $\Theta(1) = \frac{\pi}{2}$ shows us that $\theta(t_{1}) = \frac{\pi}{2}$.

For part (c), repeat the proof for part (b), except $C_{1} < 0$, and we define $\Theta: [-\pi, 0] \longrightarrow (-\infty, 0]$ instead.
\end{proof}
\hspace{0cm}

\begin{figure}[htbp]
\centerline{\includegraphics[scale=.24]{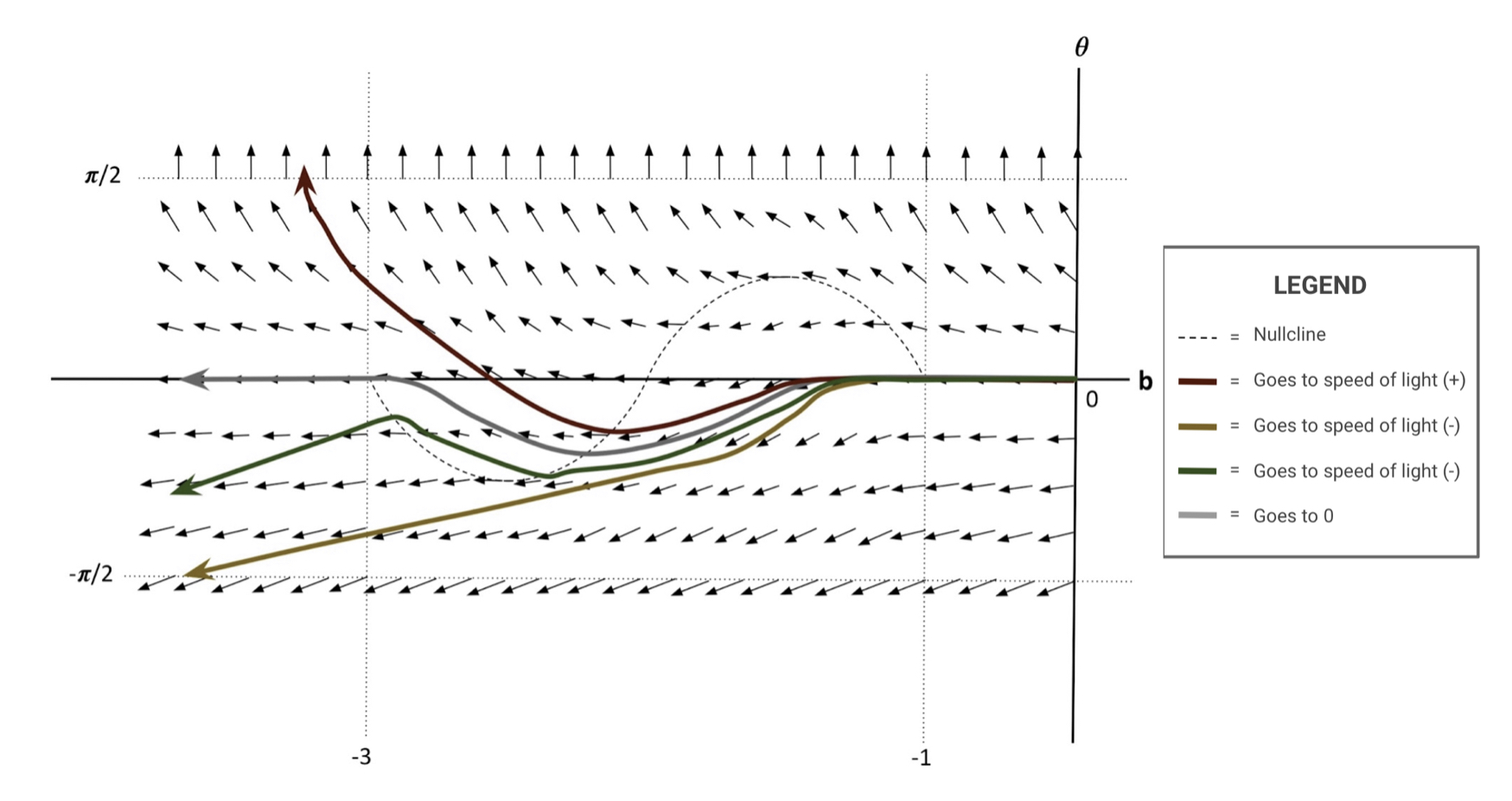}}
\caption{Hypothetical solutions to the system of ODEs}
\label{fig4}
\end{figure}

We show instability by proving the existence of a solution to (\ref{eq:sys2}) which satisfies the conditions of (c). To do this, we first work with the backward-flow defined in (\ref{eq:rev}) and make a change of variables.

\hspace{0cm}

\begin{proposition}

Assume $G(x) = \beta \sin(\pi x)\chi_{[-3,-1]}$. There exists an $\epsilon >0$ such that $0 < \beta < \epsilon$ implies that the system of ODEs in (\ref{eq:sys2}) with initial conditions at $t_{1}$:
\begin{equation}\label{eq:prop7init}
 \begin{cases} 
      b(t_{1}) = -1\\
      \theta(t_{1}) = 0\\
   \end{cases}
\end{equation}
has a unique solution $(b(t), \theta(t))$ such that at some $t_{2} > t_{1}$:
\begin{equation}
 \begin{cases} 
      b(t_{2}) = -3\\
      \theta(t_{2}) < 0. \\
   \end{cases}
\end{equation}
\end{proposition}

\begin{proof}
To start, consider the backward flow (\ref{eq:rev}) with initial conditions
\begin{equation}
 \begin{cases} 
      b(t_{1}) = -3\\
      \theta(t_{1}) = 0\\
   \end{cases}
\end{equation}
and make the following change of variables:
\begin{equation}
 \begin{cases} 
      y = \theta\\
      x = b + 2.
   \end{cases}
\end{equation}
Using the fact that $\frac{dy}{dx} = \frac{\dot{\theta}}{\dot{b}}$ gives us the following ODE:
\begin{equation}
 \begin{cases} 
    \frac{dy_{\beta}}{dx} = \frac{a}{2|m|}(1 + \sin(y_{\beta}))(\beta \sin(\pi x) - a \sec(y_{\beta})\tan(y_{\beta})) \\
    y_{\beta}(-1) = 0.
 \end{cases}
\end{equation}

By Lemma \ref{lem:1} below, there exists an $\epsilon > 0$ such that $0 < \beta < \epsilon$ implies $y_{\beta}(1) > 0$. Because of the uniqueness of solutions for first-order ODEs, a $\beta$ satisfying the previous statement implies that the solution for the forward-flow with initial conditions specified in (\ref{eq:prop7init}) ends up below $0$. This can be intuitively seen in Figure \ref{fig5}. 
\end{proof}

\hspace{0cm}

\begin{figure}[htbp]
\centerline{\includegraphics[scale=.18]{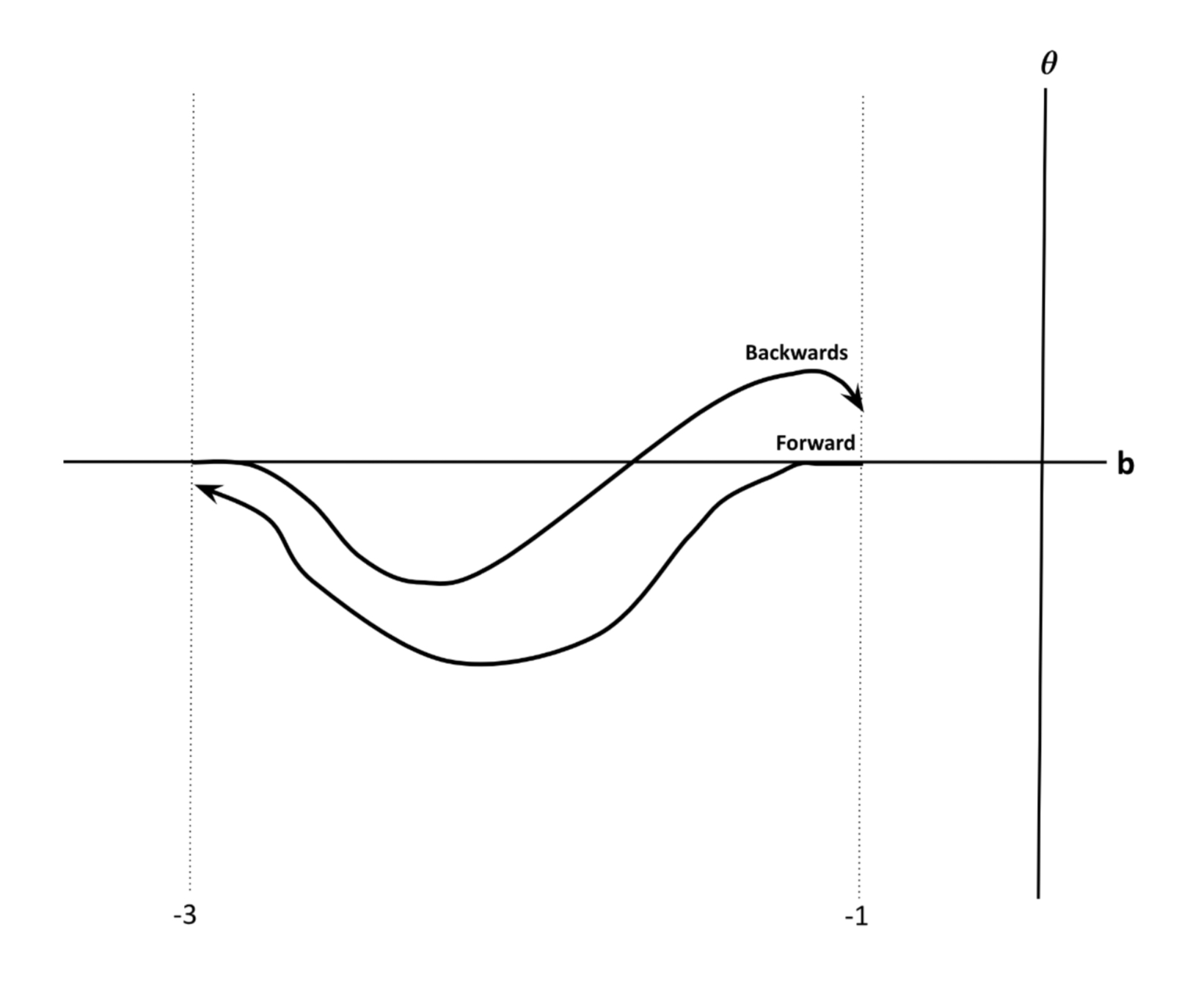}}
\caption{Relation between forward and backward solution}
\label{fig5}
\end{figure}

\begin{lemma}\label{lem:1}
Assume that $a> 0$. Suppose $y_{\beta}: \mathbb{R} \longrightarrow \mathbb{R}$ is a function satisfying: 
\begin{equation}\label{lemma1}
 \begin{cases} 
    \frac{dy_{\beta}}{dx} = \frac{a}{2|m|}(1 + \sin(y_{\beta}))(\beta \sin(\pi x) - a \sec(y_{\beta})\tan(y_{\beta})) \\
    y_{\beta}(-1) = 0.
 \end{cases}
\end{equation}
There exists an $\epsilon > 0$ such that $0 < \beta < \epsilon$ implies $y_{\beta}(1) > 0$.
\end{lemma}

\newpage{}

\begin{proof}
Consider $y(x, \beta) = y_{\beta}(x)$ and let $Z = \frac{\partial y}{\partial \beta}$. We can rewrite (\ref{lemma1}) as:

\begin{equation}
 \begin{cases} 
    \frac{\partial y}{\partial x} = \frac{a}{2|m|}(1 + \sin(y))(\beta \sin(\pi x) - a \sec(y)\tan(y)) \\
    y(-1, \beta) = 0 \\
    Z(-1, \beta) = 0.
 \end{cases}
\end{equation}

Because $y(1, 0) = 0$, it suffices to show that $Z(1,0) = \frac{\partial y}{\partial \beta}(1,0) > 0$. Using the fact that $\frac{\partial Z}{\partial x} = \frac{\partial^2 y}{\partial x \partial \beta} = \frac{\partial^2 y}{\partial \beta \partial x}$, and substituting $\beta = 0$ (meaning $y = 0$), we arrive at the following linear differential equation for $Z$:
\begin{equation}
 \begin{cases} 
    \frac{\partial Z}{\partial x}(x,0) = \frac{a}{2|m|}\sin(\pi x) - \frac{a^2}{2|m|}Z \\
    Z(-1,0) = 0.
 \end{cases}
\end{equation}

The solution to this is simply:
\begin{equation*}
Z(x,0) = \frac{a}{2|m|} e^{-\frac{a^2 x}{2|m|}} \int_{-1}^{x} \sin(\pi t)e^{\frac{a^2 t}{2|m|}} dt.
\end{equation*}
We have:
\begin{equation}
Z(1,0) = \frac{a}{2|m|} e^{-\frac{a^2}{2|m|}} \int_{-1}^{1} \sin(\pi t)e^{\frac{a^2 t}{2|m|}} dt = \frac{2\pi a |m| (1 - e^{-\frac{a^2}{|m|}})}{4\pi^2 m^2 + a^4}.
\end{equation}
Because $a$ is assumed to be positive, we have $Z(1,0) > 0$ as needed.

\end{proof}

\section{Summary and Outlook}
We have shown that the static solution to this problem, where the particle remains at rest forever, is stable for particles with positive bare mass. However, for particles with negative bare mass, the static solution is highly unstable. That is, a small amount of radiation can cause the particle to accelerate to the speed of light in finite time. Though this result is not intuitive, it is also not very surprising when considering the model we used. In the initial conditions for the wave equation, we took:

\begin{equation}
u(0,s) = -\frac{a}{2}|s| + V_{0}(s).
\end{equation}
Recall that the field energy density is $\epsilon(t,s) =T^{\mbox{\tiny field}}_{00} =  \frac{1}{2}(u_t^2+u_s^2)$. 
Therefore this initial condition has an infinite amount of energy: 
\begin{equation}
\int_{-\infty}^{\infty} \epsilon(s, 0) ds = \int_{-\infty}^{\infty} \frac{a^2}{4} + \dot{V}_0^2
 \ ds = \infty.
\end{equation}
Since the total energy of the system is conserved, there is an infinite amount of energy available that can be transferred to the particle, allowing it to accelerate to the speed of light.  But another reason this can happen in finite time is that in this model the Kiessling force $f(t,q,\dot{q})$ itself diverges as $|\dot{q}| \to 1$.  

Hence, looking forward, we would like to examine different models for the joint evolution, in which such problems are not present. In one such model, the scalar field would be governed by the Klein-Gordon equation rather than the wave equation:
\begin{equation}
    \begin{cases} 
      \partial_{t}^2 u - \partial_{s}^2 u + \mu^2u^2= a\delta(s - q(t)) \\
      u(0, s) = \frac{a}{2\mu}e^{-\mu \lvert s\rvert} + v_0(s)\\
      \partial_{t}u(0, s) = v_1(s). \\
    \end{cases}
\end{equation}
This would add a mass term to the field equations and change the part of the initial conditions that represents the static solution. In this model, the particle would start with a finite amount of energy. It would be interesting to see if a particle with negative bare mass still accelerates to the speed of light. We are currently investigating this.

Another model to consider is one in which the field equations are {\em fully} relativistic.  The wave operator appearing in (\ref{eq:wave}) is of course relativistic, but  the delta source on the right-hand side of the equation is not manifestly so.  It turns out that it is possible to modify this right-hand side so that the equation itself becomes fully relativistic. It is possible to show that for this modified equation for a massless scalar field, the Kiessling force will not diverge if the particle velocity reaches the speed of light, and stability of the static solution is restored.  This result will appear elsewhere \cite{FroLeiTah}.

Additionally, we would like to explore what would happen with two particles instead of one. Mathematically, this would involve the sum of two Dirac delta functions as the source of the wave equation. This may necessitate the use of differential-delay equations rather than simple ODEs, which would require much more intricate analysis.

\begin{section}{Acknowledgements}
We thank Vu Hoang and Maria Radosz for helping us correct the formula for the particle self-force, and Lawrence Frolov for reading through the paper and providing helpful comments.  We are grateful to the anonymous referee for many helpful suggestions and comments.

\end{section}

\providecommand{\bysame}{\leavevmode\hbox to3em{\hrulefill}\thinspace}
\providecommand{\MR}{\relax\ifhmode\unskip\space\fi MR }
\providecommand{\MRhref}[2]{%
  \href{http://www.ams.org/mathscinet-getitem?mr=#1}{#2}
}
\providecommand{\href}[2]{#2}


\begin{thebibliography}{10}

\bibitem{Bop40}
Fritz Bopp.
\newblock Eine lineare {T}heorie des {Elektrons}.
\newblock {\em Annalen der Physik}, 430(5):345--384, 1940.

\bibitem{Bop43}
Fritz Bopp.
\newblock Lineare {T}heorie des {E}lektrons. ii.
\newblock {\em Annalen der Physik}, 434(7-8):573--608, 1943.

\bibitem{Dir38}
Paul Adrien~Maurice Dirac.
\newblock Classical theory of radiating electrons.
\newblock {\em Proceedings of the Royal Society of London. Series A.
  Mathematical and Physical Sciences}, 167(929):148--169, 1938.

\bibitem{EKR09}
Yves Elskens, Michael K.-H. Kiessling, and Valeria Ricci.
\newblock The {Vlasov} limit for a system of particles which interact with a
  wave field.
\newblock {\em Communications in Mathematical Physics}, 285(2):673--712, 2009.

\bibitem{FroLeiTah}
Lawrence Frolov, Sam Leigh, and A.~Shadi Tahvildar-Zadeh.
\newblock On the joint evolution of fields and particles in one space
  dimension.
\newblock {\em in preparation}, 2023.

\bibitem{hoang}
Vu~Hoang, Maria Radosz, Angel Harb, Aaron DeLeon, and Alan Baza.
\newblock Radiation reaction in higher-order electrodynamics.
\newblock {\em Journal of Mathematical Physics}, 62(7):072901, 2021.

\bibitem{kiessling}
Michael K.-H. Kiessling.
\newblock Force on a point charge source of the classical electromagnetic
  field.
\newblock {\em Phys. Rev. D}, 100:065012, Sep 2019.

\bibitem{kiessling-err}
Michael K.-H. Kiessling.
\newblock Erratum: Force on a point charge source of the classical
  electromagnetic field [phys. rev. d 100, 065012 (2019)].
\newblock {\em Phys. Rev. D}, 101:109901, May 2020.

\bibitem{KieTah23}
Michael K.-H. Kiessling and A.~Shadi Tahvildar-Zadeh.
\newblock {Bopp-Land\'e-Thomas-Podolsky} electrodynamics as initial value
  problem.
\newblock {\em in preparation}, 2023.

\bibitem{Lan41}
Alfred Land{\'e}.
\newblock Finite self-energies in radiation theory. {P}art {I}.
\newblock {\em Physical Review}, 60(2):121, 1941.

\bibitem{LT41}
Alfred Land{\'e} and Llewellyn~H Thomas.
\newblock Finite self-energies in radiation theory. {P}art {II}.
\newblock {\em Physical Review}, 60(7):514, 1941.

\bibitem{Pod42}
Boris Podolsky.
\newblock A generalized electrodynamics part {I}—non-quantum.
\newblock {\em Physical Review}, 62(1-2):68, 1942.

\bibitem{Poi1906}
Henri Poincar{\'e}.
\newblock Sur la dynamique de l’{\'e}lectron.
\newblock {\em Rendiconti del Circolo Matematico di Palermo (1884-1940)},
  21(1):129--175, 1906.

\bibitem{SpohnBook}
Herbert Spohn.
\newblock {\em Dynamics of charged particles and their radiation field}.
\newblock Cambridge university press, 2004.

\bibitem{Weyl}
Hermann Weyl.
\newblock Feld und {M}aterie.
\newblock {\em Annalen der Physik}, 65:541--563, 1921.

\end{thebibliography}
\end{document}